\documentclass[reqno]{amsart}
\usepackage{amsmath,amssymb}
\usepackage{mathtools}
\usepackage{graphics,epsfig,color}
\usepackage[mathscr]{euscript}
\newtheorem{theorem}{Theorem}[section]
\newtheorem{proposition}[theorem]{Proposition}
\newtheorem{lemma}[theorem]{Lemma}

\theoremstyle{definition}

\numberwithin{equation}{section}
\numberwithin{theorem}{section}
\renewcommand{\epsilon}{\varepsilon}

\newcommand{\mc}[1]{{\mathcal #1}}
\newcommand{\bb}[1]{{\mathbb #1}}
\newcommand{\ms}[1]{{\mathscr #1}}
\newcommand{\bs}[1]{{\boldsymbol #1}}

\newcommand{\e}{\epsilon}
\newcommand{\ee}{\mathrm{e}}

\newcommand{\id}{{1 \mskip -5mu {\rm I}}}

\newcommand{\Ent}{\mathop{\rm Ent}\nolimits}

\newcommand{\de}{\mathop{}\!\mathrm{d}}

\title{Large deviations for a binary collision model: energy evaporation}
\author[G.\ Basile]{Giada Basile}
\address{Giada Basile \hfill\break \indent
   Dipartimento di Matematica, Universit\`a di Roma `La Sapienza'
   \hfill\break \indent
   P.le Aldo Moro 2, 00185 Roma, Italy}
 \email{basile@mat.uniroma1.it}
 \author[D.\ Benedetto]{Dario Benedetto}
 \address{Dario Benedetto \hfill\break \indent
   Dipartimento di Matematica, Universit\`a di Roma `La Sapienza'
   \hfill\break \indent
   P.le Aldo Moro 2, 00185 Roma, Italy}
 \email{benedetto@mat.uniroma1.it}
\author[L.\ Bertini]{Lorenzo Bertini}
\address{Lorenzo Bertini \hfill\break \indent
   Dipartimento di Matematica, Universit\`a di Roma `La Sapienza'
   \hfill\break \indent
   P.le Aldo Moro 2, 00185 Roma, Italy}
 \email{bertini@mat.uniroma1.it}
  \author[E.\ Caglioti]{Emanuele Caglioti}
 \address{Emanuele Caglioti \hfill\break \indent
	  Dipartimento di Matematica, Universit\`a di Roma `La Sapienza' 
	\hfill\break \indent
	P.le Aldo Moro 2, 00185 Roma, Italy}
 \email{caglioti@mat.uniroma1.it}

\begin{document}
\begin{abstract}
We analyze the large deviations for   a discrete energy Kac-like walk. In particular, we exhibit a path, with probability exponentially small   in the number of particles, that looses energy.
\end{abstract}


\maketitle
\thispagestyle{empty}

\section{Introduction}
\label{s:0}

Large deviations associated to Boltzmann-type equations have been the
object of recent investigations. The most challenging
case of Newtonian dynamics
of hard spheres in the
Boltzmann-Grad limit has
been heuristically discussed in \cite{Bou} and rigorously,
by means of cluster expansion,  in
\cite{BGSS2}.
The case of microscopic stochastic dynamics has
been originally analyzed in \cite{Le}, where
a large deviation upper bound is derived.
In \cite{Re} a large deviation principle is obtained
for a space inhomogeneous model with a finite set of velocities.
More recently, in \cite{BBBO}
it is considered a homogeneous model which conserves momentum but not
energy. The large deviation upper bound is achieved,
while the lower bound is obtained for a restricted class of paths.
In \cite{He} the analogous results are provided for the Kac walk,
which conserves also the energy.

We emphasize that, except in the case
of bounded velocities, a proof of a large deviation principle with
matching upper and lower bound is still missing,
even in the homogeneous case.
A key issue is the
possible occurrence of macroscopic paths with finite rate function
that violate the conservation of the energy.  A class of examples has
been constructed in \cite{He} by exploiting the solutions of
homogeneous Boltzmann equations provided by Lu and Wennberg in
\cite{LuW}, for which the energy is increasing.
 
Here we consider a Kac-like microscopic dynamics with discrete energy,
that is inspired by the so-called KMP model \cite{KMP,BGL}, and
described as follows: $N$ particles, with equally spaced energy
levels, evolve via random binary collisions such that in each
collision the total energy is preserved. At the kinetic level the
one-particle energy distribution evolves according to a discrete
homogeneous Boltzmann equation, that is an infinite system of coupled
ordinary differential equations.  We focus on the large deviation
properties of the pair empirical measure and flux and propose a
candidate rate function $I$ when the initial distribution of the
energies satisfies the micro-canonical constraint, i.e. the total
energy is fixed. By the arguments in \cite{BBBO, BBBC1, He}, it can be
shown that the large deviation upper bound holds with rate function
$I$, and a matching lower bound can be proven for a restricted class
of path s which conserve the energy. Our main novel point is the
construction of a path $(\bar f, \bar Q)$ which looses the energy and
whose probability is exponentially small with rate
$I(\bar f, \bar Q)$.
This result quantifies the probability to violate the conservation
of the energy, which is exponentially small in the number of particles
$N$.
It can be compared with the result in \cite{QY} where,
in the contest of the derivation of incompressible Navier-Stokes equations
from stochastic lattice gas, it is shown that the probability of violating
the incompressibility condition is of the order $\ee^{-N^2}$.


Referring to \cite{BBBC1} for a general discussion, we emphasize that,
due to the micro-canonical constraint, at the kinetic level the energy
cannot increase and the candidate rate function $I$ is different from
the one in \cite{Le, He}.


\subsection*{The model}
Given $N\ge 2$, a \emph{configuration} is defined by $N$ energies  in
$\bb N$.
The configuration space is therefore given by $\Sigma_N=\bb N^{N}$.
Elements of $\Sigma_N$ are denoted by $\bs \e\coloneqq(\e_1,.., \e_N)$ and  we denote by  $\Sigma_{N,E}$ the configuration space with total
energy  $E\in \bb N$, i.e.
\begin{equation*}
\Sigma_{N,E}\coloneqq \big\{\bs \e\in \bb N^{N}: \sum_{i=1}^N \e_i=E\big\}.
\end{equation*}  
The microscopic dynamics is defining by choosing at random a pair $\{i,j\}$ and redistributing uniformly the corresponding energies.
Therefore we consider  the Markov processes on $\Sigma_{N}$ whose generator acts on bounded functions $f\colon \Sigma_N\to \bb R$ as
\begin{equation*}
  \ms L_N f = \frac 1 N  \sum_{\{i,j\}} L_{ij} f,
\end{equation*}
where the sum is carried over the unordered pairs $\{i,j\}\subset \{1,..,N\}$, $i\neq j$, and 
\begin{equation}\label{Lij}
  L_{ij}f \, (\bs \e) = \frac 1 {\e_i +\e_j+1}
  \sum_{\ell =0}^{\e_i+\e_j}
  \big[ f(T_{ij}^{\ell} \bs \e) - f (\bs \e) \big],
\end{equation}
in which
\begin{equation*}
  \big(T_{ij}^{\ell} \bs \e\big)_k \coloneqq
  \begin{cases}
    \ell  & \textrm{if } k=i \\
    \e_i+\e_j-\ell & \textrm{if } k=j \\
    \e_k & \textrm{otherwise.} 
  \end{cases}
\end{equation*}
Observe that, for each $E>0$, the  dynamics preserves $\Sigma_{N,E}$. Moreover, it is ergodic when restricted to $\Sigma_{N,E}$ and reversible with respect to the uniform measure
on $\Sigma_{N,E}$.

We denote by $(\bs \e(t))_{t\ge 0}$ the continuous time Markov chain generated by
$\ms L_N$.  In particular, the path $\bs \e(\cdot)$ is
piecewise constant, and the transition probability
of its jumps is given in \eqref{Lij}.
Fix hereafter $T>0$. Given a probability $\nu$ on $\Sigma_{N,E}$ we denote
by $\bb P_\nu^N$ the law of this chain on the time interval $[0,T]$,
when the initial datum is sampled according to $\nu$.
Observe that $\bb P_\nu^N$ is a probability on the Skorokhod
space $D([0,T];\Sigma_{N,E})$.
As usual if $\nu=\delta_{\bs \e}$ for some $\bs \e \in \Sigma_{N,E}$, the
corresponding law is simply denoted by $\bb P_{\bs \e}^N$.
We refer to \cite{Li} for a gentle introduction to 
continuous time Markov chains.

\subsection*{Empirical observables}
Given $e \in (0, +\infty)$, we denote by $\ms P_e(\bb N)$ the set of probability measures $\pi$ on $\bb N$ with mean bounded by $e$, i.e.
such that $\sum_{\e} \e \pi(\e)\leq e$.
We consider $\ms P_e(\bb N)$ as a closed subset of the space of probability measure on $\bb N$  equipped with the weak topology.
Then  $\ms P_e(\bb N)$ endowed with the relative  topology is a compact  Polish space. Indeed, $\ms P_e(\bb N)$ is the weak closure of the set of probabilities on $\bb N$ with mean $e$.
The {\it empirical measure} records the energy of the particles
forgetting their labels, for 
$E=\lfloor N\, e\rfloor$ it is defined as  the map
$\pi^N\colon \Sigma_{N,E} \to \ms P_e(\bb N)$ given by 
\begin{equation}
  \label{1}
  \pi^N(\bs \e) \coloneqq\frac 1N \sum_{i=1}^N\delta_{\e_i}.
\end{equation}

Let $D\big([0,T]; \ms P_e(\bb N)\big)$ the set of
$\ms P_e(\bb N)$-valued c{\'a}dl{\'a}g paths endowed with the
Skorokhod topology and the corresponding Borel $\sigma$-algebra.
With a slight abuse of notation we denote also by $\pi^N$ the map 
from $D\big([0,T]; \Sigma_{N,E} \big)$ to $D\big([0,T]; \ms P_e(\bb N)\big)$ defined by $\pi^N_t(\bs \e)\coloneqq\pi^N(\bs \e(t))$,
$t\in [0,T]$. 

We also introduce the {\it empirical flow} that records
the collisions of the particles forgetting their labels.
To this end, 
we denote by $\ms M$ the  subset of the finite measures $Q$  on
$[0,T]\times \bb N^2\times \bb N^2$  that satisfy
$Q(\de t; \e,\e_*, \e',\e_*')=Q(\de t; \e_*,\e, \e',\e_*') =Q(\de t; \e, \e_*, \e'_*,\e')$. We endow  $\ms M$ with the weak*
topology and the associated Borel $\sigma$-algebra.
The empirical flow is the map
$Q^N\colon D\big([0,T]; \Sigma_{N,E} \big) \to \ms M$
defined by
\begin{equation}
  \label{2}
  Q^N(\bs \e) (F) \coloneqq\frac 1N
  \sum_{\{i,j\}} \sum_{k\ge 1} F\big(\tau^{i,j}_k;
  \e_i(\tau^{i,j}_k-),\e_j({\tau^{i,j}_k}-),
  \e_i(\tau^{i,j}_k),\e_j(\tau^{i,j}_k)\big) 
  \quad 
\end{equation}
where $F\colon [0,T]\times \bb N^2\times \bb N^2\to \bb R$
is continuous, bounded, and satisfies $F(t; \e,\e_*, \e',\e_*') =F(t; \e_*,\e, \e',\e_*') = F(t; \e,\e_*, \e'_*,\e')$,
while $(\tau^{i,j}_k)_{k\ge 1}$ are the
jump times of the pair $(\e_i,\e_j)$. Here, $\e_i(t-) = \lim_{s\uparrow t} \e_i(s)$.
In view of the conservation of the energy, the
measure $Q^N(\de t;\cdot)$ is supported on 
$\ms E \coloneqq\{\e+\e_*=\e'+\e_*'\}\subset \bb N^2\times \bb N^2$.

For each $\bs \e \in \Sigma_{N,E}$, with $\bb P^N_{\bs \e}$ probability
one, the pair $(\pi^N,Q^N)$ satisfies the following
balance equation that express the conservation of probability. For each
$\phi\colon [0,T]\times \bb N\to \bb R$ bounded and
continuously differentiable with respect to time 
\begin{equation}
  \label{bal1}
  \begin{split}
  &\pi^N_T(\phi_T)-\pi^N_0(\phi_0)-\int_0^T\!\de t\, 
  \pi^N_t(\partial_t \phi_t)\\
  &\quad 
  +\int_0^T\!\!\sum_{ \e,\e_*, \e',\e_*'} Q^N(\de t; \e,\e_*, \e',\e_*')\big[ \phi_t(\e)+\phi_t(\e_*)
  -\phi_t(\e')-\phi_t(\e_*')\big] =0.
  \end{split}
\end{equation}

\subsection*{Law of large numbers}
Fix $m\in \ms P(\bb N)$ and assume one of the following condition:
$m$ is a point mass or
the support of $m$ does not generate a proper sub-lattice of
$\bb Z$.
Note that in the second case $m$ satisfies the condition
for the local central limit theorem for i.i.d. lattice random
variables, see \cite[\S VII.1]{Pe}. 
For $\gamma\in \bb R$ we set
\begin{equation}
  Z_\gamma=Z_\gamma(m) \coloneqq\sum_{\e} m(\e) \ee^{\gamma \e}.
\end{equation}
We assume that there exists $\gamma^*\in (0, +\infty]$ such that
$Z_\gamma < +\infty$ for $\gamma\in (0, \gamma^*)$ and
$Z_\gamma \uparrow + \infty$ for $\gamma \uparrow \gamma^*$.
For $e\in (0, +\infty)$, we
then define the probability $\mu_{N, e}$ on
$\Sigma_{N,\lfloor N e \rfloor}$ by considering
i.i.d. $m$-distributed energies and conditioning to the total
energy, i.e.
\begin{equation}\label{df:muN}
\mu_{N,e}\coloneqq m^{\otimes N}\big(\cdot \vert \sum_{i=1}^N \e_i = \lfloor N e \rfloor  \big),
\end{equation}  
that will be chosen as the initial distribution of the microscopic dynamics.
In the case of point mass, we require that $\{e\}$ is  exactly
the support of $m$.
Observe that, by the equivalence of the ensembles, as $N\to +\infty$
the one-marginal of
$\mu_{N,e}$ converges to the probability $m_e$
given by
\begin{equation}
  \label{eq:me}
  m_e (\e)\coloneqq \frac  {\ee^{\gamma_e \e} m(\e)}{Z_{\gamma_e}},
  \text{ where }\gamma_e<\gamma^* \text{ is such that }
  \sum_\e \e m_e(\e) = e.
  \end{equation}

Denoting by $B$ the collision kernel in \eqref{Lij}, i.e.
\begin{equation}
  \label{eq:B}
  B(\e,\e_*,\e',\e'_*)
  =\frac 1 {\e +\e_* +1} \id_{\{\e +\e_*=\e' +\e'_*\}} \id_{\big\{\{\e, \e_*\}\neq \{\e', \e'_*\}\big\}},
\end{equation}
the law of the large numbers for the empirical measure is described  by the following discrete homogeneous Boltzmann equation
\begin{equation}\label{BE}
\partial_t f_t(\e)=\sum_{\e_*, \e', \e'_*}B(\e,\e_*,\e',\e'_*)\big[f_t(\e') f_t(\e'_*)- f_t(\e)f_t(\e_*) \big].
\end{equation}
More precisely, in probability with respect to $\bb P^N_{\mu_{N,e}}$,
the empirical path $\big(\pi^N_t\big)_{t\in [0,T]}$ converges to
$\big (f_t\big)_{t\in [0,T]}$ where $f_t$ is the unique solution of
the Cauchy problem associated to \eqref{BE} with initial datum
$f_0=m_e$.
As the proof of this statement can be achieved by adapting the chaos
propagation arguments in \cite{Sn}, we omit the details.
The law of the large
numbers of the empirical flow $Q^N$ reads
\begin{equation}
  Q^N(\de t; \e,\e_*, \e',\e_*')\longrightarrow \frac 12
  \de t\,f_t(\e)f_t(\e_*)B(\e,\e_*,\e',\e'_*),
\end{equation}
where the convergence is   in probability with respect to $\bb P^N_{\mu_{N,e}}$. We refer to Lemma \ref{lemma1} below for the proof.

In the general contest of homogeneous Boltzmann equations, uniqueness
of the Cauchy problem associated to \eqref{BE} holds for paths $f_t$
that conserves the energy, see e.g. \cite{MW}.  However in the present
case, since the
$\sup_{\e, \e_*}\sum_{\e', \e'_*} B(\e, \e_*, \e', \e'_*)$ is bounded,
Gronwall's inequality implies the uniqueness without assuming the
energy conservation, see e.g. Lemma 4.1 in \cite{BBBO}.  In
particular, by  uniqueness,
for this model Lu and Wennberg like solutions do not
exist. We finally observe that \eqref{BE} admits a one-parameter family
of stationary solutions given by $f_{\textrm{stat}}(\e)= p(1-p)^{\e}$,
$p\in (0,1]$.

\subsection*{The candidate rate function}
For $e\in (0, +\infty)$, let $\ms S_e$ be the (closed) subset of
$D\big([0,T]; \ms P_e(\bb N)\big)\times \ms M$ given by elements
$(\pi,Q)$ that satisfies the balance equation
\begin{equation}
  \label{bal}
  \begin{split}
  &\pi_T(\phi_T)-\pi_0(\phi_0)-\int_0^T\!\de t\, 
  \pi_t(\partial_t \phi_t)\\
  &\qquad 
  +\int_0 ^T \sum_{\e, \e_*, \e', \e'_*} Q(\de t;\e, \e_*, \e', \e'_*)\big[ \phi_t(\e)+\phi_t(\e_*)
  -\phi_t(\e')-\phi_t(\e_*')\big] =0
  \end{split}
\end{equation}
for each $\phi:[0,T]\times \bb N\to \bb R$  bounded and continuously differentiable in $t$.
We consider $\ms S_e$ endowed with the relative topology and the
corresponding Borel $\sigma$-algebra.

For $\pi\in D\big([0,T]; \ms P_e(\bb N)\big)$ let $Q^\pi$ be
the measure defined by
\begin{equation}
  \label{4}
  Q^\pi(\de t; \e, \e_*, \e', \e'_*) \coloneqq \frac 12
  \de t \, \pi_t(\e) \pi_t(\e_*)   B(\e,\e_*,\e',\e'_*),
\end{equation}
where $B$ is the collision kernel in  \eqref{eq:B}.
Observe that $Q^\pi(\de t,\cdot)$ is supported on $\ms E$.
  Let $\ms S_{e}^{\mathrm{ac}}$ be the subset of $\ms S_e$ given by the elements
  $(\pi,Q)$ such that 
 $\pi\in C\big([0,T];\ms P_e(\bb N)\big)$ and $Q\ll Q^\pi$.
The dynamical rate function $J\colon \ms S_e \to [0,+\infty]$ is defined by
\begin{equation}
  \label{5}
  J(\pi,Q)\coloneqq
  \begin{cases}
    {\displaystyle 
    \int_0^T \sum_{\e, \e_*, \e', \e'_*}  \de Q^\pi \Big[ 
    \, \frac{\de Q\phantom{^\pi}}{\de Q^\pi} \log \frac{\de Q\phantom{^\pi}}{\de Q^\pi} -
    \frac{\de Q\phantom{^\pi}}{\de Q^\pi}  + 1\Big]  } & \textrm{if } (\pi,Q)\in \ms
    S_e ^{\mathrm{ac}}\\ \\
    + \infty& \textrm{otherwise } 
  \end{cases}
\end{equation}

Given two probabilities $\mu_1, \mu_2$, the relative entropy
$\Ent(\mu_2\vert\mu_1)$ is defined as
$\Ent(\mu_2\vert\mu_1)=\int d\mu_1 \rho\log\rho$, where
$d\mu_2=\rho\, d\mu_1$, understanding that $\Ent(\mu_2|\mu_1)=+\infty$
if $\mu_2$ is not absolutely continuous with respect to $\mu_1$. Let
$H_e\colon \ms P_e (\bb N) \to [0, +\infty]$ be defined by
\begin{equation}\label{def:H}
  H_e (\pi)= \Ent(\pi\vert m_e) +(\gamma^*-\gamma_e) \Big [e-\sum_{\e}
  \e \,\pi(\e) \Big], 
\end{equation} 
where $m_e$ and $\gamma_e$ are as in \eqref{eq:me}.  When $m$ is the
point mass on $e$, $H_e(\pi)$ is zero when $\pi = \delta_e$ and
$+\infty$ otherwise.  The candidate large deviation rate function is
given by
\begin{equation}\label{ihj}
  I (\pi,Q) \coloneqq H_e(\pi_0) + J(\pi,Q).  
\end{equation}
As discussed in \cite{BBBC1}, the (static) large deviations of the
empirical measure with respect to the probability $\mu_{N,e}$ are
described by the rate function $H_e$, where the second term on the
r.h.s. of \eqref{def:H} is the cost of having energy less than $e$.
Note that if $\gamma^*=+\infty$, then $H_e(\pi)$ is finite only if
the energy of $\pi$ is $e$.
A key ingredient in the proof is the local central limit
theorem for the sum of independent $m_e$ distributed random variables.

Denote by $\hat {\ms S}_e$ the subset of $\ms S$ given by the pair
$(\pi, Q)$ such that
$$\int_0^T \sum_{\e,\e_*,\e',\e'_*}
Q(\de t, \e,\e_*,\e',\e'_*) (\e + \e_*) < +\infty.$$
If $(\pi, Q)\in \hat {\ms S}_e$, the balance equation \eqref{bal}
implies that the path $\pi_t$ conserves the energy.

As already mentioned, a proof of a large deviations principle for the
pair $(\pi^N,Q^N)$ with matching upper and lower has not been yet
achieved. The analysis in \cite{BBBO,He} implies however the large
deviation upper bound with rate $I$ with a matching lower bound
on the set $\hat {\ms S}_e$. The precise statement is the following.

\begin{theorem}
  \label{upperbound}
  Fix $e \in (0,+\infty)$ and let $\mu_{N,e}$ be the family of
  probabilities on $\Sigma_{N, \lfloor Ne \rfloor}$ defined in
  \eqref{df:muN}.
  The family $\{\bb P^N_{\mu_{N,e}} \circ (\pi^N,Q^N)^{-1}\}$ satisfies a
  large deviations upper bound with good rate function
  $I\colon \ms S_e \to [0,+\infty]$, namely $I$ has compact level sets and
  for each closed $\mc{C} \subset \ms S$
  \begin{equation}\label{upeq}
    \varlimsup_{N\to +\infty} \frac 1N \log   \bb P^N_{\mu_{N,e}}
    \Big( (\pi^N,Q^N)\in \mc{C}
    \Big) \le - \inf_{(\pi,Q)\in \mc C} \, I(\pi,Q).
  \end{equation}
  Moreover, for  each open  $\mc{O} \subset \ms S_e$
  \begin{equation}\label{lbeq}
    \varliminf_{N\to +\infty} \frac 1N \log   \bb P^N_{\mu_{N,e}}
    \Big( (\pi^N,Q^N)\in \mc O
    \Big) \ge - \inf_{(\pi,Q)\in \mc O\cap \hat {\ms S}_e } \, I(\pi,Q).
  \end{equation}
\end{theorem}

Referring to \cite{BBBO,He} for comments on the technicalities
involved in the lower bound, we now turn to the novel point of the
present analysis, that is the construction of paths $(\pi,Q)$ -- with
$\pi$ not energy conserving -- whose probability is precisely of the
order $\exp\{-N I(\pi,Q)\}$. Since these paths do not belong to
$\hat{\ms S}_e$, this result provides insights on the large deviations
properties of Kac's walk not covered by \eqref{lbeq}.
As the large deviations upper bound is already covered by
\eqref{upeq}, we focus on the matching lower bound. 

\begin{theorem}\label{main}
  Fix $e \in (0,+\infty)$ and let $\mu_{N,e}$ be the family of
  probabilities on $\Sigma_{N, \lfloor Ne \rfloor}$ defined in
  \eqref{df:muN}.
  For each $t^*\in (0, T)$ there exists a path
  $(\bar f, \bar Q)$, satisfying $\sum_{\e} \bar f_t(\e) \e =e$ for
  $t\in [0, t^*)$ and $\sum_{\e} \bar f_t(\e) \e < e$ for
  $t\in [t^*, T]$, such that $I(\bar f, \bar Q)< +\infty$ and
  \begin{equation}\label{2.16}
    \varliminf_{N\to\infty} \frac 1 N \log \bb
    P^N_{\mu_{N,e}}\big((\pi^N, Q^N) \in \mathcal O \big)\geq - I(\bar
    f, \bar Q), 
  \end{equation} 
 for any open neighborhood $\mathcal O\ni (\bar f, \bar Q)$.
\end{theorem}

We will provide a self-contained proof of this statement that do not
rely on Theorem~\ref{upperbound}.
In the argument we take
advantage of the fact that the energies are in $\bb N$. However we
expect that the strategy can be extended also to the continuous case.
In Section \ref{sez:perturbed} we
construct a path $(\bar f, \bar Q)$ satisfying the above requirements,
i.e. with evaporating energy for $t>t^*$ and such that
$I(\bar f, \bar Q)< +\infty$.  The lower bound \eqref{2.16} is then
proven in Section \ref{sez:ld}.
For the sake of concreteness, the proposed path $(\bar f, \bar Q)$ has
zero energy for $t\in(t_*, T]$.

\section{Perturbed Boltzmann equation}
\label{sez:perturbed}
Fix $t^*\in (0,T)$. In order to construct a path $(\bar f, \bar Q)$, satisfying $\sum_{\e} \bar f_t(\e) \e =e$ for
  $t\in [0, t^*)$ and $\sum_{\e} \bar f_t(\e) \e < e$ for
  $t\in (t^*, T]$, we start by considering a solution to 
a perturbed Boltzmann equation, namely a Boltzmann equation with a suitable modified collision kernel.

Consider the  collision kernel $\tilde B$ given by
\begin{equation}\label{tB}
  \tilde B(\e, \e_*, \e', \e'_*)=
  \frac 1 2  \delta_{\e, \e_*}\delta_{\e +\e_*, \e'+\e'_*} \big[\delta_{\e', \e+\e_*} + \delta_{\e'_*, \e + \e_*} \big] \id_{\big\{\{\e,\e_*\}\neq \{\e',\e_*'\}\big\}}, 
\end{equation}
that describes a scenario in which only particles with the same energy collide, and in each collision the whole energy is transferred to a single particle.
The Cauchy problem for the corresponding modified homogeneous Boltzmann equation reads
\begin{equation}\label{mBE}
  \begin{cases}\vspace{0.1 cm}
    \displaystyle \partial_t f_t(\e)=\sum_{\e_*, \e', \e'_*} \big[\tilde B (\e', \e'_*, \e, \e_* ) f_t(\e') f_t(\e_*')- \tilde B(\e, \e_*, \e', \e'_*) f_t(\e)f_t(\e_*) \big],\\
    f_0 (\e) = m(\e).
  \end{cases}
 \end{equation} 

\begin{proposition}\label{3.1}
  Assume that $m$ has energy $e$, and 
  let $f$ be the unique solution to \eqref{mBE}.
  Then its energy is
  conserved, i.e. for any
  $t\in [0, +\infty)$ $\sum_{\e \ge 1} f_t(\e)\e =e$, while $f_t$
  weakly converges to $\delta_0$,
  as $t\to +\infty$. Moreover, for every $t\geq 0$,
  \begin{equation}
    \label{eq:stimafe}
    \begin{aligned}
      i)&\ \     f_t(\e) \leq \frac 2 {1+t}, \ \ \text{ for } \e \ge 1,\\
      ii)&\ \   \sum_{\e \ge 1} f_t(\e) \le \frac c{\sqrt{1+t}},\\
      iii)&\ \   \sum_{\e \ge 1} f_t(\e) \log \e \le c\frac 1 {\sqrt{1+t}}
      \big( 1 + \log(1+t)\big)
    \end{aligned}
  \end{equation}
  where $c=c(e)$  does not depend on $t$ and the initial datum $m$.
\end{proposition}  

\begin{proof}
  We prove eq.s \eqref{eq:stimafe}, from which the convergence of $f_t$ follows.
  The modified Boltzmann equation reads as
  \begin{equation}
    \begin{aligned}
    \dot f_t(0)&= \frac 1 2 \sum_{\e \ge 1}f_t(\e)^2,  \\
    \dot f_t(\e) &=  -f_t(\e)^2\ \
    &\text{for } \e\ge 1 \text{ odd},\\
    \dot f_t(\e) &=  \frac 1 2 f_t(\e/2)^2 -f_t(\e)^2\ \ \ \ 
    &\text{for } \e\ge 2 \text{ even}.\\
   \end{aligned}
\end{equation}
Note that the equation for $f_t(\e)$ involves only $f_t(\e')$ with $\e'\le \e$,
then the system has global and unique solution.
If $\e$ is odd,
$$f_t(\e)  = \frac {f_0(\e)}{1+tf_0(\e)} \le \frac 1{1+t}.$$
If $\e$ is even, set $\xi_t(\e) = (1+t) f_t(\e)$.
Let $T_0$ be the first time $t$ 
such that $\xi_t(\e') =2$ for some $\e'\le \e$.
The time $T_0$ is strictly
positive, since $\xi_0(\e') \le 1$ for any $\e'$.
For $t<T_0$ it holds
$$\dot \xi_t(\e) = \frac 1{1+t} \left( \xi_t(\e) - \xi_t^2(\e) + \frac 12
  \xi_t^2(\e/2)\right) < \xi_t(\e) - \xi_t^2(\e) + 2\le 3(2-\xi_t(\e)),$$
then $T_0=+\infty$,  and this concludes the proof
of $i)$ in \eqref{eq:stimafe}.

In order to prove $ii)$ and $iii)$  we first note that
\begin{equation}
  \label{eq:energia}
  \sum_0^{2^n}\e f_t(\e) =
  \sum_0^{2^n}\e f_0(\e)
  -\int_0^t \de s \sum_{2^{n-1}+1}^{2^n} \e f_s(\e)^2,
\end{equation}
and then  $\sum_0^{+\infty} \e f_t(\e) \le e$.
Inequality $ii)$ and $iii)$ follows by using this fact
and the Chebyshev's inequality.
We conclude the proof by noticing that the energy is in fact conserved,
since  $f_t(\e) \le e/\e$
and then for any $h\ge 1$,
$$\sum_{h}^{+\infty} \e f_t(\e)^2 \le e \max_{\e\ge h} f_t(\e) \le
\frac {e^2}{h},$$
which assures that the right-hand-side of  eq. \eqref{eq:energia},
is vanishing as $n\to +\infty$
\end{proof}

For the modified Boltzmann equation with rate in
\eqref{tB}  the energy vanishes for dispersion to infinity as
$t\to +\infty$.
We reparametrize the time  so that
this happens at a finite time.
Fixed $t^*\in (0,T)$, let $\alpha \colon [0, t^*)\to [0, +\infty)$ given
by $\alpha(t)=\frac t {1-t/t^*}$.
Letting $f$ the solution to \eqref{mBE}, set
\begin{equation}
  \bar f_t (\e) =
  \begin{cases}
    f_{\alpha(t)}(\e) & t\in[0, t^*)\\
      \delta_{\e, 0}  & t\in [t^*, T],
  \end{cases}  
\end{equation}
which satisfies the homogeneous Boltzmann equation with time dependent collision kernel
  \begin{equation}\label{Bbar}
    \bar B_t(\e, \e_*, \e', \e'_* )= \begin{cases}
      \dot\alpha(t) \tilde B (\e, \e_*, \e', \e'_* )
      & t\in [0, t^*)\\
        0 & t\in [t^*, T].
      \end{cases}
    \end{equation}  
    We define the corresponding flux $d\bar Q = \de t \, \bar q_t$, where 
    \begin{equation}
      \bar q_t(\e, \e_*, \e', \e'_*)= \frac 12
      \bar f_{t}(\e)\bar f_{t}(\e_*) \tilde B_t (\e, \e_*, \e', \e'_* ),
    \end{equation}
    so that the pair $(\bar f, \bar Q)$ satisfies the balance equation
    \eqref{bal}. Observe that, by construction,
    $\sum_{\e} \bar f_t(\e) \e =e$ for $t\in [0, t^*)$ and
    $\sum_{\e} \bar f_t(\e) \e =0$ for $t\in [t^*, T]$.
    We now show
    that the pair $(\bar f, \bar Q)$ is such that
    $I(\bar f, \bar Q)<+\infty$.  Since $\bar f_0 =m$, it is enough to
    show $J(\bar f, \bar Q)<+\infty$. This is stated in the next
    Proposition.

\begin{proposition}\label{prop2}
  For $(\bar f, \bar Q)$ defined above the dynamical rate function
  $J(\bar f, \bar Q)$ is finite.
\end{proposition}

 \begin{proof}
   Since $(\bar f, \bar Q)\in \ms
   S_e ^{\mathrm{ac}}$, the dynamical rate function defined in eq. \eqref{5}
   is given by
   \begin{equation}
     \label{eq:Jdastimare}
     \frac 12 \int_0^{t^*} \de t
     \sum_{\e,\e_*,\e',\e_*'}
     \bar f_t(\e) \bar f_t(\e_*)\bar B_t   \left(  \log \frac {\bar B_t}{B}-
       1\right) +
   \frac 12 \int_0^T  \bar f_t(\e) \bar f_t(\e_*) B.
   \end{equation}
   For $t<t^*$
   we have  that
   $$\bar B_t \left( \log  \frac {\bar B_t}{B} -1 \right) 
   = \dot  \alpha \tilde B \left( \log \dot \alpha + \log
     \frac {1+2\e}2  - 1\right).$$
   Since $\log \dot \alpha = 2 \log ( 1 + \alpha/t^*)$,
   the first integral is
   $$\frac 1 2 \int_0^{+\infty}d\alpha \sum_{\e \ge 1}
   f_\alpha^2(\e) \left( 2\log \left( 1+\frac \alpha{t^*}\right)  + \log
     \frac {1+2\e}2  - 1\right).$$
   Using \eqref{eq:stimafe}
   we bound this term by
   $$c\int_0^{+\infty} \frac 1{(1+\alpha)^{3/2}}
   (1+\log (1+\alpha))<+\infty,$$
   where $c$ depend only on $e$ and $t^*$.

   The second integral in eq. \eqref{eq:Jdastimare}
   is 
   $$\frac 12 \int_0^{t^*} \de t \left(\sum_{\e\ge 1} f_{\alpha(t)}\right)^2
   + \frac 12 (T-t^*) \le \frac T2,$$
   which completes the proof.
 \end{proof}

\section{Large deviations lower bound}
\label{sez:ld}
In order to explain the strategy to prove \eqref{2.16}, we recall
some basic facts on the large deviations lower bound.
Let $\{P_n\}$ be a sequence of probabilities on a topological space
$\ms X$. Fix $x\in \ms X$ and a open neighborhood $\mathcal O\ni
x$. To obtain a lower bound for $P_n(\mc O)$ we modify the probability
$P_n$ so that $x$ becomes the typical behavior. If we are able to do
so by paying -- as measured by the relative entropy with respect to
$P_n$ -- not too much then we obtain a good lower bound. The precise
statement is summarized in the next lemma, see e.g.\ \cite{Mar} for
its proof.

\begin{lemma}
  \label{lemmalungo}
  Let $\{P_n\}_{n\in \bb N}$ be a sequence of probabilities 
  on a completely regular topological space $\ms X$
  and fix $x\in \ms X$.  Assume that there exists a sequence 
  $\{P_n^x\}$ weakly convergent to $\delta_x$ and such that
  \begin{equation}
    \label{entb}
    \varlimsup_{n\to\infty} \frac 1n \Ent\big(P_n^x \big|
    P_n\big)
    \le I(x)
  \end{equation}
  for some $I\colon \ms X\to [0,+\infty]$.
  Then for any open neighborhood $\mathcal O\ni x$
  \begin{equation*}
    \varliminf_{n\to\infty} \frac 1 n \log P_n (\mathcal O)\geq - I(x).
  \end{equation*}
\end{lemma}

In most of the applications, and indeed also in our case, the strategy
suggested by the above lemma is implemented together with a density
argument. The family of perturbed probabilities $P_n^x$ is not
constructed for the point $x$ itself but rather for an approximation
$x_k$; if the function $I$ is continuous along the sequence $x_k$ then
this will do as well. We emphasize that in typical infinite
dimensional applications -- as in the present case -- the rate
function $I$ is only lower semicontinuous so the sequence $x_k$ has to
be properly chosen.
We summarize the argument in the next statement which is deduced from
Lemma~\ref{lemmalungo} by a straightforward diagonal argument.

\begin{lemma}
  \label{lemmacorto}
  Let $\{P_n\}_{n\in \bb N}$ be a sequence of probabilities on a
  completely regular topological space $\ms X$, fix $x\in \ms X$, and
  a sequence $x_k\to x$.  Assume that there exists
  $I\colon \ms X\to [0,+\infty]$ meeting the following conditions:
  \begin{itemize}
  \item [(i)] for each $k\in \bb N$ there exists a family
    $\{P_n^{x_k}\}_{n\in\bb N}$
    satisfying the conditions in Lemma~\ref{lemmalungo};
  \item [(ii)] $\varlimsup_k\, I(x_k) \le I(x)$.
  \end{itemize}
  Then, for any open  neighborhoods $\mathcal O\ni x$
  \begin{equation*}
    \varliminf_{n\to\infty} \frac 1 n \log P_n (\mathcal O)\geq - I(x).
  \end{equation*}
\end{lemma}
To implement condition {\it (i)} and {\it (ii)} in
the previous lemma, 
for $0<\delta <t_*$,
define the pair $(\bar f^\delta, \bar Q^\delta)$ by
 \begin{equation}\label{fl}
\bar f^\delta_t (\e) =
  \begin{cases}
    \bar f_t(\e) & t\in[0, t_*-\delta)\\
    \bar f_{t_*-\delta}(\e)  & t\in [t^*-\delta, T],
  \end{cases}  
 \end{equation}  
 and $d\bar Q^\delta= \de t \,
 \bar q^\delta_t$ with $\bar q^\delta_t = \bar q_t\id_{[0, t_*-\delta)}(t)$.
 Let $\mu_{N,e}$ as in the statement of Theorem
 \ref{main}, 
 and denote by $\bar{\bb P}^{N,\delta}_{\mu_{N,e}}$ the law of the
 microscopic dynamics with the perturbed collision kernel
 $\bar B^\delta= \bar B \id_{[0, t_*-\delta)}(t) $, $\bar B$ in
 \eqref{Bbar}. Then the following two Lemmata imply the large
 deviation lower bound \eqref{2.16}.
 \begin{lemma}\label{lemma1}
For each $\delta\in (0, t_*)$ as $N\to +\infty$  the pair $(\pi^N, Q^N)$ converges in $\bar{\bb P}^{N,\delta}_{\mu_{N,e}}$ probability  to $(\bar f^\delta, \bar Q^\delta)$. Furthermore,
\begin{equation}
  \displaystyle \lim_{N\to +\infty} \frac 1 N\Ent (\bar{\bb P}^{N,\delta}_{\mu_{N,e}}\vert \bb P^N_{\mu_{N,e}}) = I(\bar f^\delta, \bar Q^\delta).
  \end{equation}
  \end{lemma}

  \begin{lemma}\label{lemma2}
    As $\delta\downarrow 0$ we have  $(\bar f^\delta, \bar Q^\delta)\to (\bar f, \bar Q)$ and
    $I(\bar f^\delta, \bar Q^\delta)\to I(\bar f, \bar Q)$.
  \end{lemma}  

  \begin{proof}[Proof of Lemma \ref{lemma1}]
    By definition of $\bar B^\delta$, $\sup_{t}\sup_{\e,\e'}\sum_{\e', \e'_*}\bar B_t^\delta(\e, \e', \e'_*)\leq c_\delta$.
    Therefore, by classical chaos propagation argument,
    $\pi^N$ converges in $\bar{\bb P}^{N,\delta}_{\mu_{N,e}}$ probability to $\bar f^\delta$. To deduce the convergence of the empirical  flow, it is enough to observe that for each bounded $F_t(\e, \e_*, \e', \e'_*)$
    \begin{equation}\label{m1}\begin{split}
        M^F_t\coloneqq & \int_0^t\sum_{\e, \e_*, \e', \e'_*} Q^N(\de s, \e, \e_*, \e', \e'_*)F_s(\e, \e_*, \e', \e'_*)\\
        &-\frac 1 2 \int_0^t \de s \sum_{\e, \e_*, \e', \e'_*} \pi_s^N(\e)\pi_s^N(\e_*)\bar B_s^\delta(\e, \e_*, \e', \e'_*)F_s(\e, \e_*, \e', \e'_*)\\
        &+\frac 1 2 \frac 1 N\int_0^t \de s \sum_{\e, \e', \e'_*} \pi_s^N(\e)\bar B_s^\delta(\e, \e, \e', \e'_*)F_s(\e, \e, \e', \e'_*)
   \end{split} \end{equation}  
    is a $\bar{\bb P}^{N,\delta}_{\mu_{N,e}}$ martingale with predictable quadratic variation
    \begin{equation*}\begin{split}
\langle M^F \rangle_t = & \frac 1 2 \frac 1 N \int_0^t \de s \sum_{\e, \e_*, \e', \e'_*} \pi_s^N(\e)\pi_s^N(\e_*)\bar B_s^\delta(\e, \e_*, \e', \e'_*)F^2_s(\e, \e_*, \e', \e'_*)\\
         &-\frac 1 2 \frac 1 {N^2}\int_0^t \de s \sum_{\e, \e', \e'_*} \pi_s^N(\e)\bar B_s^\delta(\e, \e, \e', \e'_*)F^2_s(\e, \e, \e', \e'_*).
    \end{split}\end{equation*}  
Set $F^\delta_t(\e, \e_*, \e', \e'_*)=\log(\bar B^\delta_t/B)$. By standard Markov chain computation, the relative entropy  of $\bar{\bb P} ^{N,\delta}_{\mu_{N,e}}$ with respect to $\bb P ^N_{\mu_{N,e}}$ is given by
\begin{equation}\begin{split}\label{ent}
  & \frac 1 N \, \Ent (\bar{\bb P}^{N,\delta}_{\mu_{N,e}}\vert \bb P^N_{\mu_{N,e}})
  \\ & =\bar{\bb E}^{N,\delta}_{\mu_{N,e}}\Big (Q^N (F^\delta)-\frac 1 2\int_0^T\!\! \de t\, \sum_{\e, \e_*, \e', \e'_*}\pi_t^N(\e) \pi_t^N(\e_*)\big[\bar B^\delta_t(\e, \e_*, \e', \e'_*) -B(\e, \e_*, \e', \e'_*) \big]                    \\
  & +\frac 1 N  \frac 1 2\int_0^T\!\! \de t\, \sum_{\e, \e', \e'_*}\pi_t^N(\e)\big[\bar B^\delta_t(\e, \e, \e', \e'_*) -B(\e, \e, \e', \e'_*) \big]\Big).
\end{split}\end{equation}
Since $\sup_{t}\sup_{\e,\e'}\sum_{\e', \e'_*}\bar B_t^\delta(\e, \e', \e'_*)\leq c_\delta$, by the law of large numbers
\begin{equation*}\begin{split}
 & \lim_{N\to\infty}\bar{\bb E}^{N,\delta}_{\mu_{N,e}}\Big (\frac 1 2\int_0^T\!\! \de t\, \sum_{\e, \e_*, \e', \e'_*}\pi_t^N(\e) \pi_t^N(\e_*)\big[\bar B^\delta_t(\e, \e_*, \e', \e'_*) -B(\e, \e_*, \e', \e'_*) \big]             \Big)\\
 & =\frac 1 2\int_0^T\!\! \de t\, \sum_{\e, \e_*, \e', \e'_*}\bar f_t^\delta(\e) \bar f_t^\delta(\e_*)\big[\bar B^\delta_t(\e, \e_*, \e', \e'_*) -B(\e, \e_*, \e', \e'_*) \big],            
 \end{split}\end{equation*} 
while the last term on the right hand side of \eqref{ent} vanishes as $N$ diverges.
Again, by the law of large numbers, in order to prove $\bar{\bb E}^{N,\delta}_{\mu_{N,e}}\Big (Q^N (F^\delta)\Big) \to \bar Q ^\delta(F^\delta)$ it is enough to show the uniform integrability of $Q^N(F^\delta)$ with respect to $\bar{\bb P}^{N,\delta}_{\mu_{N,e}}$. By exploiting the martingale decomposition \eqref{m1}, since $|F^\delta|\leq c_\delta(1+\log(1+\e+\e_*))$, a direct computation yields
$\bar{\bb E}^{N,\delta}_{\mu_{N,e}}\Big( Q^N(F^\delta)^2\Big) \leq c_\delta$, which implies the requested uniform integrability.
Observing that  $H_e(\bar f^\delta_0)=0$, and recalling \eqref{5} and \eqref{ihj}, the proof is concluded.
  \end{proof}

  \begin{proof}[Proof of Lemma \ref{lemma2}]
  The convergence of $(\bar f^\delta, \bar Q^\delta)$ to $(\bar f, \bar Q)$ follows from \eqref{fl}, the continuity of $t\mapsto \bar f_t$, and the integrability of $\bar q_t$. The convergence of the rate function is achieved by the arguments in the proof of Proposition \ref{prop2} and dominated convergence.
    
  \end{proof}


\begin{thebibliography}{99}

  

\bibitem{BBB} Basile G., Benedetto D., Bertini, L.; {\it A gradient flow approach to linear Boltzmann equation}, Ann. Sc. Norm. Super. Pisa Cl. Sci. (5) Vol. XXI, 955-987, 2020.



\bibitem{BBBC1} Basile G., Benedetto D., Bertini, L., Caglioti, E.:
  {\it On large deviations associated to
    homogeneous Boltzmann equations: an improved upper bound}
  in preparation (2021)


\bibitem{BBBO} Basile G., Benedetto D., Bertini, L., Orrieri C.;
  {\it Large Deviations for Kac-Like Walks},
J Stat Phys 184, 10 (2021). https://doi.org/10.1007/s10955-021-02794-2




  
\bibitem{BGL} Bertini L., Gabrielli D., Lebowitz  J.L.;
  \emph{Large deviations for a stochastic model of heat
flow} J. Stat. Phys., 121(5):843–885, 2005.



  
  
  
\bibitem{BGSS2}
  Bodineau T., Gallagher I., Saint-Raymond L., Simonella S.; {\it Statistical dynamics of a hard sphere gas: fluctuating Boltzmann equation and large deviations}, preprint, arXiv:2008.10403

\bibitem{Bou}  Bouchet F.; {\it Is the Boltzmann Equation Reversible? A Large Deviation Perspective on the Irreversibility Paradox}, J. Stat. Phys., 2020.











\bibitem{He} Heydecker D.:
  {\it
    Large Deviations of Kac's Conservative Particle System and Energy Non-Conserving Solutions to the Boltzmann Equation: A Counterexample to the Predicted Rate Function}  arXiv:2103.14550 (2021)


  


  
\bibitem{KMP} Kipnis C., Marchioro C., Presutti E.; {\it Heat flow in an exactly solvable model}, J. Stat. Phys. 27, 65–74, 1982.

\bibitem{Le}L\'eonard, C.; {\it On large deviations for particle systems associated with spatially homogeneous Boltzmann type equations}, Probab. Th. Rel. Fields 101, 1–44, 1995.


\bibitem{Li} Liggett T.M.;
  Continuous Time Markov Processes: An Introduction,
  American Mathematical Society, Providence, 2010
  ISBN: 978-0-8218-4949-1 

    

  
 %
  %

\bibitem{LuW} Lu X., Wennberg B.,
  {\it Solutions with increasing energy for the spatially homogeneous Boltzmann equation}, Nonlinear Analysis: Real World Applications
  (3) 2 pp. 243-–258 https://doi.org/10.1016/S1468-1218(01)00026-8 (2002)

\bibitem{Mar} Mariani M.; {\it A $\Gamma$-convergence approach to large deviations}, Ann. Sc. Norm. Super. Pisa Cl. Sci. 18, 951-976, 2018.


  \bibitem{MW} Mischler S., Wennberg B.;{\it On the spatially homogeneous Boltzmann equation}, Ann. Inst. Henri Poincar\'e,
Analyse non lin\'eaire, 16 (4), 467-501, 1999.

\bibitem{Pe} Petrov V.V.; Sums of independent random variables, Springer-Verlag, New York, 1975. ISBN 978-3-540-06635-4



\bibitem{QY} Quastel J., Yau H-T.;  
  {\it Lattice gases, large deviations, and the incompressible Navier-Stokes equations}, Annals of Math. 148, 51-108, 1998.

\bibitem{Re}  Rezakhanlou F.; {\it Large deviations from a kinetic limit},  Annals of Prob. 26(3), 1259–1340, 1998.




\bibitem{Sn} Sznitman A.S.; {\it Topics in propagation of chaos}, in Hennequin PL. (eds) Ecole d'Et\'e de Probabilit\'es de Saint-Flour XIX - 1989, Lect. Notes Math., vol 1464, Springer Berlin Heidelberg, 1991  







\end{thebibliography}
\end{document}